\documentclass[conference, balanced]{IEEEtran}
\IEEEoverridecommandlockouts
\usepackage{mathtools}
\usepackage{amsmath,amssymb,amsthm,mathrsfs}
\usepackage{dsfont}
\usepackage{psfrag}
\usepackage{algorithmicx}
\usepackage{algorithm}
\usepackage{algpseudocode}
\usepackage{nccmath}
\usepackage{bm}
\usepackage{color}
\usepackage{cite}
\usepackage{stfloats}
\usepackage{float}
\usepackage{lipsum}
\usepackage{multicol}
\usepackage{hyperref}
\usepackage{pst-poly,pst-node,pst-plot,pst-grad}
\usepackage{multirow}
\usepackage{cite,url}
\usepackage{enumerate}

\usepackage[enable-write18, interaction=nonstopmode]{auto-pst-pdf}
\usepackage{ifplatform, pst-pdf, xkeyval}
\usepackage[T1]{fontenc}

\usepackage{lipsum}

\newcounter{MYtempeqncnt}
\allowdisplaybreaks[2]

\theoremstyle{definition}
\newtheorem{definition}{Definition}
\newtheorem{theorem}{Theorem}
\newtheorem{lem}{Lemma}
\newtheorem{corol}{Corollary}

\theoremstyle{remark}
\newtheorem{remark}{Remark}

\algnewcommand{\algorithmicand}{\textbf{ and }}
\algnewcommand{\algorithmicor}{\textbf{ or }}
\algnewcommand{\OR}{\algorithmicor}
\algnewcommand{\AND}{\algorithmicand}

\title{The Optimal Power Control Policy for an Energy Harvesting System with Look-Ahead: Bernoulli Energy Arrivals}

\author{Ali~Zibaeenejad,~\IEEEmembership{Member,~IEEE}, and~Jun~Chen
\thanks{The authors are with electrical and computer engineering department, McMaster University,
Hamilton, ON L8S 4L8, Canada (e-mail: \{azibaeen, chenjun\}@mcmaster.ca). Also, the first author is with
the school of electrical and computer engineering, Shiraz University, Shiraz, Fars 71348-51154, Iran (e-mail: zibaeenejad@shirazu.ac.ir).}}

\begin{document}

\maketitle

\begin{abstract}
We study power control for an energy harvesting  communication system with independent and identically distributed Bernoulli energy arrivals. It is assumed that the transmitter is equipped with a finite-sized rechargeable battery and is able to look ahead to observe a fixed number of future arrivals. A complete characterization is provided for the optimal power control policy that achieves the maximum long-term average throughput over an additive white Gaussian noise channel.

\end{abstract}

\section{Introduction}
Supplying required energy of a communication system by energy harvesting (EH) from natural energy resources is not only beneficial from the environmental perspective, but also essential for long-lasting self-sustainable affordable telecommunication which can be employed in places with no electricity infrastructure. On the other hand, the EH systems require to handle related challenges, such as varying nature of green energy resources and limited battery storage capacity.\\
\indent Consider an EH communication system with a transmitter (TX) and a receiver (RX), connected by an additive white Gaussian noise (AWGN) channel.
The TX is equipped with a rechargeable battery of a given storage size, and is capable of harvesting the energy arrivals,  which are assumed to be independent and identically distributed (i.i.d.). The communication session consists of $T$ (discrete) time slots, and the instantaneous rate achieved at each time slot is a function of its allocated energy. A power control policy specifies the energy assignment across the time slots according to the initial battery energy level, the harvested energy sequence thus far, and the knowledge of future harvested energy arrivals. The reward associated with each policy is the average throughput over $T-$horizon. The average throughput optimization (ATO) problem aims to determine the optimal policy that achieves the maximum average throughput. \\
\indent In the seminal paper~\cite{ulukus}, the authors studied the ATO problem over a finite horizon for the offline model, where the energy arrivals are non-causally known at the TX, and they derived the optimal policy for this model based on ~\cite{modiano}. The analyses of the offline model for more general channels can be found in~\cite{yang2012, tutuncuoglu2012, yener, zhang12, review15} (and the references therein). In general, the optimal policies for the offline model strive to allocate the energy across the time horizon as uniformly as possible while trying to avoid energy loss due to the battery overflow.\\
\indent In landmark paper~\cite{ozgur16}, the authors studied the ATO problem over an infinite horizon for the online model, where the energy arrivals are causally known at the TX. They determined the optimal policy for Bernoulli energy arrivals and established the approximate optimality of the fixed fraction policy for general energy arrivals. Similar results were derived in ~\cite{Ulukus2017-ISIT, arafa2018} for a general concave and monotonically increasing utility function. In papers~\cite{zib-IWCIT, zib-IEMCON}, the authors studied the same problem with unlimited size battery and developed three simple online optimal policies which are optimal for the offline case as well.\\
\indent In this paper, we study the setup where the TX is able to look ahead to observe a window of size $w$ of future energy arrivals. In fact, the online model and offline model correspond to the extreme cases $w=0$ and $w=\infty$, respectively. Therefore, our formulation provides a link between these two models, which have been largely studied in isolation. From the practical perspective, the TX often has a good estimation of the amount of available energy in near future, because such energy is already harvested but not yet converted to a usable form. We investigate the ATO problem over an infinite horizon for this new setup. Specifically, we focus on Bernoulli energy arrivals and characterize the corresponding optimal policy. The main difference between this optimal policy and that of~\cite{ozgur16, Ulukus2017-ISIT, arafa2018} is as follows.
 For the new policy, if no energy arrival is seen in the look-ahead window, the battery always keeps some energy for future and spends a portion of available energy in the current time slot. In contrast, for the policy in \cite{ozgur16, Ulukus2017-ISIT, arafa2018}, the energy is only allocated to a fixed number of time slots after each battery charge and no energy is expended beyond that as the battery becomes depleted.\\
\indent The organization of this paper is as follows. In Section~\ref{Sec:Problem-Def}, we introduce the model and problem. In Section~\ref{Sec:Strategy}, we develop the structure of an optimal policy and justify the problem solving strategy. In Section~\ref{Sec:Properties}, we establish the main results and completely characterize the optimal policy. 
In Section~\ref{Sec:Conclusion}, we finally conclude this paper.

\section{Problem Definitions} \label{Sec:Problem-Def}
\indent The notations of this paper are as follows. $\mathds{N}$ and $\mathds{R}^+$ represent the set of natural numbers and set of positive real numbers, respectively. Random variables are denoted by capital letters and their realizations are written in lower case letters. The functions are denoted by calligraphic font. $\mathbb{E}$ is reserved for the expectation. The logarithms are in base~2.\\
\indent Consider a point-to-point AWGN quasi-static fading channel from a TX to an RX, where the channel gain is a constant $\gamma \in \mathds{R}^+$ for the entire communication session.
The communication is discrete time with time slot $\tau \in \mathds{N}$ as formulated by $Y_{\tau} = \sqrt{\gamma} X_{\tau} + Z_{\tau}$,
where $X_{\tau}$ and $Y_{\tau}$ are the transmitted signal and the received signal, respectively; $Z_{\tau}$ is the Gaussian noise with zero mean and unit variance. The TX is capable of harvesting energy from the environment and is equipped with a rechargeable battery with finite size $B>0$. The exogenous (harvested) energy arrivals are assumed to be i.i.d. process with known marginal distribution $P_{E}$. In this work, unless specified otherwise, we assume that $P_{E}$ is Bernoulli$(p)$-$\{0, B\}$, where $0 < p < 1$, defined as
\begin{equation} \label{eq:EH-Seq-distribtion}
  P_E(e) = \left\{
             \begin{array}{ll}
               p  &: e = B, \\
               1-p &: e = 0.
             \end{array}
           \right.
\end{equation}
Denote the energy level stored in the battery by random process $\{B_\tau, \tau \in \mathds{N}\}$  with initial level, $B_1=\beta$, where $0\leq \beta \leq B$. If the energy arrives at some time instant $\tau\in \mathds{N}$, $E_\tau = B$, and the battery is fully charged to $B_\tau = B$. In this case, any remained energy in the battery is overflowed and wasted away. If no energy arrives, $E_\tau = 0$, and the battery energy level is not escalated at time slot $\tau$.\\
\indent In this paper, we assume that the TX is able to look ahead with a fixed window size $w \in \mathds{N}$: the realization of the energy arrival sequence $\{E_t\}_{t=1}^{\tau+w}$ is known to the TX at time $\tau$.\\
\indent  The TX sends energy $A_\tau$ as an \emph{action} at time slot $\tau$, where the energy is determined by a (randomized) \emph{action function}
\begin{equation}\label{def:action-function-general}
\begin{array}{c}
  A_\tau = \mathcal{A}_\tau({(E_t)_{t=1}^{\tau+w}}, B_1)\:,\\
\text{subject to $A_{\tau} \leq B_{\tau}$}
\end{array}
\end{equation}
and gain throughput $\mathcal{R}_\tau(A_\tau) = \frac{1}{2}\log(1+ \gamma A_\tau)$, as the \emph{reward} of time $\tau$. 
Then, the battery energy level becomes
\begin{equation} \label{eq:battery-change-relation}
  B_\tau = \min\{B_{\tau-1} - A_{\tau-1} + E_\tau , B\}\:.
\end{equation}
\indent A look-ahead policy $\pi(w)$ is characterized by sequence of action functions $\{\mathcal{A}_\tau\}_{\tau=1}^\infty$. For a fixed communication session time $T \in \mathds{N}$, $\pi(w)$ gains the average (expected) throughput over $T$-horizon
\begin{equation} \label{eq:def-average-throughput}
\Gamma^{\pi(w)}_T \triangleq \frac{1}{T} \mathbb{E}\left(\sum\limits_{\tau=1}^T \frac{1}{2}\log(1+ \gamma \mathcal{A}_\tau((E_t)_{t=1}^{\tau+w}, B_1))\right)\:,
\end{equation}
as its associated reward, where the expectation is over all energy arrival sequences $\{e_t\}_{t=1}^{T}$.
\begin{definition} \label{def:long-term-throughput}
The largest average reward (channel throughput) over infinite horizon (long term) is defined as
\begin{equation} \label{eq:long-term-throughput}
  \Gamma^*_{B_1} \triangleq \sup_{\pi(w)} \liminf_{T\rightarrow \infty} \Gamma^{\pi(w)}_T \:.
\end{equation}
If $\Gamma^*_{B_1}$ is attainable by a policy $\pi^*(w)$, it is called optimal.
\end{definition}
\begin{remark}
It can be shown~\cite[Appendix~B]{ozgur16} that $\Gamma^*_{B_1}$ does not depend on $B_1$. Hence, we can drop the subscript $B_1$ in~\eqref{eq:long-term-throughput}, and assume $B_1=B$ without loss of generality (WLOG).
\end{remark}
\indent In this paper, we seek $\Gamma^*$ and the corresponding optimal policy $\pi^*(w)$ according to Definition~\ref{def:long-term-throughput}.

\section{Problem Solving Strategy}\label{Sec:Strategy}
\indent First, assume that the distribution of the harvested energy, $P_E$, is arbitrary. In general, $A_\tau$ depends on $(B_1, \{E_t\}_{t=1}^{\tau+w})$. According to~\cite{Puterman, survey}, there is no loss of optimality in~\eqref{eq:long-term-throughput} if the supremum is taken over deterministic Markovian stationary policies which only rely on system state
\begin{equation} \label{eq:state}
  S_\tau = (B_\tau, E_{\tau+1}, E_{\tau+2}, \ldots, E_{\tau+w})\:.
\end{equation}
Indeed, \eqref{eq:long-term-throughput} is attainable by an optimal stationary policy $\pi^*(w)$. Given $S_\tau$ with finite length $w+1$, knowing energy arrivals $\{E_t\}_{t=1}^{\tau}$ does not enhance $\Gamma^*$. Note that the action $A_\tau$ is not only determined by the current energy level $B_\tau$, but also it can be affected by the observed future energy arrivals within the look ahead window. As the optimal policy is Markovian and stationary, the action function~\eqref{def:action-function-general} can be simplified to the time-invariant function
\begin{equation}\label{def:action-function-stationary}
\begin{array}{c}
  A_\tau = \mathcal{A}(S_\tau)\:.\\
\text{subject to $A_{\tau} \leq B_{\tau}$}
\end{array}
\end{equation}
\indent Now, focus on Bernoulli distribution as defined in~\eqref{eq:EH-Seq-distribtion}. In this case, the state~\eqref{eq:state} can be simplified as follows. Let random variable $D_\tau$ be the time \emph{distance} of the \emph{earliest} energy arrival located inside the look-ahead window. Specifically, define
\begin{equation*}
  D_\tau \triangleq \left\{
             \begin{array}{ll}
               0 \quad\text{: if }E_{\tau+t} = 0 \text{ for all}\:  t\in\{1, \ldots, w\},\\
               \min\{t: 1\leq t\leq w, E_{t+\tau}=B\} \quad\text{: O.W.}\:.
             \end{array}
           \right.
\end{equation*}
For any given energy arrival sequence $\{e_t\}_{t=1}^{\tau+w}$ and battery level $b_\tau$ at time $\tau \in \mathds{N}$, if an energy arrival is observed at time $\tau$, i.e., $d_\tau>0$, then the optimal policy uniformly assigns instant battery energy $b_\tau$ to the following $d_\tau$ time spots. This is due to the concavity of the reward function~\eqref{eq:def-average-throughput}. Otherwise, $\mathcal{A}(b_\tau, 0, \cdots, 0)$ is the action at time $\tau$, which will be determined in the sequel. Therefore, the action function~\eqref{def:action-function-stationary} of the stationary optimal policy $\pi^*(w)$ is given by
\begin{equation} \label{eq:optimal-policy-Bern}
a_\tau = \mathcal{A}(b_\tau, e_{\tau+1}, \ldots, e_{\tau+w}) = \left\{
             \begin{array}{ll}
               \frac{b_\tau}{d_\tau} &\text{: }d_\tau \ne 0\:,\\
               \mathcal{A}(b_\tau, 0, \cdots, 0) &\text{: O.W.}\:.
             \end{array}
           \right.
\end{equation}
From~\eqref{eq:optimal-policy-Bern}, we conclude that $a_\tau$ (and so the associated reward) can be uniquely determined by $(b_\tau, d_\tau)$. Hence, the system state for Bernoulli energy arrival can be reduced to
\begin{equation} \label{eq:state-Bernoulli}
S_\tau = (B_\tau, D_\tau)\:.
\end{equation}
\indent A non-negative sequence $\{x_j\}_{j=1}^N$ with length $N\in \mathds{N}$ is called \emph{admissible}, if $\sum_{j=1}^N x_j \leq B$.
Let $b_1=B$. Define admissible sequence $\{\xi^*_j\}_{j=1}^\infty$ associated with $\pi^*(w)$ by
\begin{equation}\label{def:xi-Sequence}
  \xi^*_j \triangleq \mathcal{A}(b_j, 0, \cdots, 0)\: : j\in \mathds{N},
\end{equation}
where $b_{j+1} = b_j - \xi^*_j$.
Due to~\eqref{eq:optimal-policy-Bern}, if the battery is charged up at some time $\tau=\tau_1$, ($b_{\tau_1} = B$), but no arrival occurs later, ($e_\tau = 0$\: : $\tau>\tau_1$), $\pi^*(w)$ sends $a_{\tau} = \xi^*_{\tau-\tau_1+1}$ for $\tau \geq \tau_1$. In the sequel, $\{\xi_i^*\}_{i=1}^\infty$ and its properties are investigated. Once $\{\xi_i^*\}_{i=1}^\infty$ is determined, $\pi^*(w)$ follows Algorithm~\ref{Fig:Algorithm}.\\
\begin{algorithm}[t]
  \caption{Optimal policy $\pi^*(w)$ for Bernoulli Arrivals}\label{Fig:Algorithm}
  \begin{algorithmic}
 \Require{Window size $w$, battery capacity $B$, arrival energy sequence $\{E_t\}_{t=1}^{\tau + w}$ at any time $\tau$, and sequence $\{\xi_j^*\}_{j=1}^\infty$.}
 \Ensure{The optimal assigned energy (action) $a^*_\tau$ at time $\tau$}.
  \textit{Initialize}: \State{Set time $\tau \leftarrow 1$, next observation distance $d \leftarrow 0$, battery level $b \leftarrow B$, counter $i\leftarrow 1$, and counter  $j\leftarrow 1$}.
  \Repeat   \Comment{No observation $(d=0)$ by default}
  \While {($d = 0$ \AND $i \leq w$)}  
  \If {($e_{i+\tau} = B$)} \Comment{Finds earliest arrival distance}
  \State $d \leftarrow i$; $j\leftarrow 1$; $i\leftarrow 1$
  \EndIf
  \State $i \leftarrow i+1$
  \EndWhile
  \If {($d \ne 0$)} \Comment{If an arrival is observed}
  \State $a_\tau \leftarrow \frac{b}{d}$; $d \leftarrow d-1$
  \Else \Comment{If no arrival is observed}
  \State $a_\tau \leftarrow \xi^*_j$; $j \leftarrow j+1$; $i \leftarrow w$
  \EndIf
  \State $b \leftarrow \min\{b - a_\tau + e_{\tau+1}, B\}$ \Comment{Battery level is updated}
  \State $\tau \leftarrow \tau + 1$
  \Until
 \end{algorithmic}
 \end{algorithm}
\indent The times in interval $[\tau_1, \tau_2)$ with property $E_{\tau_1}=E_{\tau_2}=B$ and $E_{t}=0$ for $\tau_1 < t < \tau_2$ is called a \emph{cycle} with \emph{start time} $\tau_1$. Let define random variable $L = \tau_2-\tau_1$ as the (duration) \emph{time of the cycle}. Due to~\eqref{eq:EH-Seq-distribtion}, the distribution of $L$ is Geometric, i.e., $Pr\{L=k\} = p(1-p)^{k-1}$ for $k\in \mathds{N}$, with mean
\begin{equation}
\mathbb{E}(L) = \frac{1}{p}. \label{eq:E(L)}
\end{equation}
\indent If a stationary policy is employed, the processes of $\{B_\tau, \tau\in \mathds{N}\}$, $\{S_\tau, \tau\in \mathds{N}\}$, and $\{A_\tau, \tau\in \mathds{N}\}$ are non-delayed regenerative processes~\cite[Section~7.5]{Ross_prob} with cycles of time $L$: when the battery charges up to $B$ at some cycle start time $\tau_1$, memories of the processes are reset. Consequently, $\{(B_t, S_t, A_t)\}_{t=\tau_1}^\infty$ does not statistically depend on $(\{B_t\}_{t=1}^{\tau_1-1},\{S_t\}_{t=1}^{\tau_1-1}, \{A_t\}_{t=1}^{\tau_1-1})$ and time $\tau_1$.\\
\indent As $E(L)$ in~\eqref{eq:E(L)} and $|\mathcal{R}_\tau(A_\tau)|$\ are bounded and $\{A_\tau, \tau\in \mathds{N}\}$ is a non-delayed regenerative process, the ``renewal reward theorem''~\cite[Section~7.4]{Ross_prob} can be utilized to simplify the long-term average throughput achieved by a policy $\pi(w)$.
\begin{eqnarray}
  \liminf_{T\rightarrow \infty} \Gamma^{\pi(w)}_T &=& \frac{\mathbb{E}\left(\sum_{t=1}^{L} \frac{1}{2}\log(1+\gamma A_t) \right)}{\mathbb{E}(L)}\nonumber\\
    &=& \frac{p}{2} (\sum_{k=1}^\infty \sum_{j=1}^k \log(1+ \gamma Q_{j}) Pr\{L=k\}), \label{eq:renewal-2-Theta-infty}
\end{eqnarray}
where~\eqref{eq:renewal-2-Theta-infty} is due to~\eqref{eq:E(L)} and $Q_{j}$ is the energy assigned to time $j^{th}$ of a cycle conditioned on $k$.
The following definition is helpful to calculate \eqref{eq:renewal-2-Theta-infty}.
\begin{definition} \label{def:Theta_inf}
Let $\{x_j\}_{j=1}^\infty$ be an admissible sequence. Define
\begin{align}
\mathcal{T}_\infty(\{x_j\}_{i=1}^\infty) \triangleq& \sum\limits_{k=1}^w p^2(1-p)^{k-1}\frac{k}{2}\log(1+\gamma\frac{B}{k}) \nonumber \\
&+\sum\limits_{j=1}^\infty p(1-p)^{j+w-1}\frac{1}{2}\log(1+\gamma x_j) \nonumber\\
+\sum\limits_{k=1}^\infty &p^2(1-p)^{k+w-1}\frac{w}{2}\log(1+\gamma\frac{B-\sum_{j=1}^{k} x_j}{w})\:.
\label{optimize-target-Infinity-modify}
\end{align}
\end{definition}

\begin{lem} \label{Lem:optimal-energy-general-equation}
The long-term average throughput~\eqref{eq:def-average-throughput} of optimal policy~$\pi^*(w)$ with associated sequence $\{\xi_i^*\}_{i=1}^\infty$ satisfies
\begin{equation}\label{eq:optimal-energy-general-equation}
\Gamma^* = \mathcal{T}_\infty(\{\xi_i^*\}_{i=1}^\infty)\:.
\end{equation}
\end{lem}
\begin{proof}
The proof follows from Definition~\eqref{def:long-term-throughput},~\eqref{eq:renewal-2-Theta-infty} and Definition~\ref{def:Theta_inf}. First, energy assignments $Q_j$ in~\eqref{eq:renewal-2-Theta-infty} for $\pi^*(w)$ is determined in the following. Assume a new cycle is started at time $\tau = \tau_1$ and so the battery energy level is $B$. Given the cycle time $L=k$,  the following two cases can be considered.
\begin{enumerate}[a.]
  \item Case $k \leq w$: In this case, $D_\tau = k$ as the arrival is observed in the window. Hence, according to~\eqref{eq:optimal-policy-Bern}, we have
      \begin{equation} \label{eq:Lem:optimal-energy-general-equation:1}
        Q_{j} = \frac{B}{k}\quad \text{for }j\in\{1,\ldots,k\}.
      \end{equation}
  \item Case $k > w$: In this case, the TX allocates $\{\xi_j^*\}_{j=1}^{k-w}$ to the first $k-w$ time slots of each cycle; then, it uniformly distributes the remained energy of the battery $(B- \sum_{i=1}^{k-w} Q_i)$ into time slots $k-w+1, \ldots, k$ due to~\eqref{eq:optimal-policy-Bern}, as soon as the arrival is observable (look-ahead window covers time $k$).
      \begin{equation} \label{eq:Lem:optimal-energy-general-equation:2}
      Q_j = \left\{
      \begin{array}{ll}
        \xi_j^* &:j\in \{1, \ldots, k-w\},\\
        \frac{B-\sum_{j=1}^{k-w}\xi_j^*}{w} &:j\in \{k-w+1, \ldots, w\}\:.
      \end{array}
      \right.
      \end{equation}
      Finally, the proof can be concluded from the following calculation of \eqref{eq:renewal-2-Theta-infty} for $\pi^*(w)$ based on
      \eqref{eq:Lem:optimal-energy-general-equation:1} and \eqref{eq:Lem:optimal-energy-general-equation:2}.
   \begin{eqnarray}
   \displaystyle
   \Gamma^* &=& p\sum\limits_{k=1}^w p(1-p)^{k-1}\frac{k}{2}\log(1+\gamma\frac{B}{k})\nonumber\\
   &&+p\sum\limits_{k=w+1}^\infty
    p(1-p)^{k-1}[\sum\limits_{j=1}^{k-w}\frac{1}{2}\log(1+\gamma \xi_j^*) \nonumber\\
    &&\qquad \qquad+\frac{w}{2}\log(1+\gamma\frac{B-\sum_{j=1}^{k-w}\xi_j^*}{w})]\nonumber\\
    &=& \sum\limits_{k=1}^w p^2(1-p)^{k-1}\frac{k}{2}\log(1+\gamma\frac{B}{k})\nonumber\\
    &&+\sum\limits_{j=1}^\infty \frac{1}{2}\log(1+\gamma \xi_j^*)( \sum\limits_{k=j+w}^\infty
    p^2(1-p)^{k-1})\nonumber\\
    &&+\sum\limits_{k=w+1}^\infty
    p^2(1-p)^{k-1}\frac{w}{2}\log(1+\gamma\frac{B-\sum_{j=1}^{k-w}\xi_j^*}{w})\:. \nonumber\\
    &=& \sum\limits_{k=1}^w p^2(1-p)^{k-1}\frac{k}{2}\log(1+\gamma\frac{B}{k})\nonumber\\
    &&+\sum\limits_{j=1}^\infty p(1-p)^{j+w-1}\frac{1}{2}\log(1+\gamma \xi_j^*)\nonumber\\
    &&+\sum\limits_{k=1}^\infty
    p^2(1-p)^{k+w-1}\frac{w}{2}\log(1+\gamma\frac{B-\sum_{j=1}^{k}\xi_j^*}{w})\:. \nonumber\\
    \label{optimize-target-Infinity}
    \end{eqnarray}

\end{enumerate}
\end{proof}
\section{Properties of the Optimal Policy} \label{Sec:Properties}
In this section, we characterize the energy sequence $\{\xi_j^*\}_{j=1}^\infty$ and its properties. The main result of this paper is as follows.
\begin{theorem} \label{Theorem:Main-Result}
Define $\mathcal{T}_\infty^* = \sup\mathcal{T}_{\infty}(x_1, x_2, \cdots)$, where the supremum is over all admissible sequences $\{x_j\}_{j=1}^\infty$.
Then, the maximum long-term average reward (channel throughput) of the look-ahead model is given by
\begin{equation} \label{eq:Theorem:Average-Throughput}
  \Gamma^* = \mathcal{T}_\infty^*\:.
\end{equation}
Moreover, $\{\xi_j^*\}_{j=1}^\infty$ induced by the optimal policy is the unique maximizer of $\mathcal{T}_\infty$ and is the unique sequence satisfying
\begin{subequations}
\begin{eqnarray}
    \frac{1-p}{1+\gamma \xi_{j+1}^*} &=& \frac{1}{1+\gamma \xi_{j}^*} - \frac{p}{1+\frac{\gamma}{w}(B-\sum_{i=1}^{j} \xi_i^*)}, \label{eq:E_j^*_recursive}\\
    \sum\limits_{j=1}^{\infty} \xi_j^* &=& B\:. \label{eq:energy_constraint-infinit}
\end{eqnarray}
\end{subequations}
It is also a strictly decreasing positive sequence with property
  \begin{equation}\label{eq:theorem:property}
     \xi_j^* < \frac{B- \sum_{i=1}^{j} \xi_i^*}{w} \:.
  \end{equation}
\end{theorem}
\begin{remark}
The optimal average throughput of the studied model coincides with the optimal average throughput of the non-causal model if $w=\infty$ is set in~\eqref{optimize-target-Infinity}. In this case the optimal average throughput is given by
\begin{equation*}
\Gamma^* = \sum_{k=1}^{\infty} p^2(1-p)^{k-1}\frac{k}{2}\log(1+\gamma \frac{B}{k})\:.
\end{equation*}
\end{remark}
\indent In Fig.~\ref{fig:Throughput-window}, we illustrate the (long-term) average throughput $\Gamma^*$ as a function of window size $w$ using~\eqref{eq:Theorem:Average-Throughput} for the given system parameters. Although the optimal throughput rate is an increasing function of $w$, if the window size $w\geq 5$ in Fig.~\ref{fig:Throughput-window}, this communication system achieves the optimal average throughput corresponds to $w=\infty$ within a gap smaller than $0.5\%$. The rest of this section is devoted to the proof of theorem~\ref{Theorem:Main-Result}.\\
\indent The proof of~\eqref{eq:Theorem:Average-Throughput} is as follows. As mentioned, there exists a stationary policy $\pi^*(w)$ which attains $\Gamma^*$~\cite{Puterman, survey}. That policy also satisfies~\eqref{eq:optimal-energy-general-equation} in which $\{\xi_j^*\}_{j=1}^\infty$ is associated with $\pi^*(w)$. As $\pi^*(w)$ is optimal, $\{\xi_j^*\}_{j=1}^\infty$ must maximize $\mathcal{T}_{\infty}$ in \eqref{eq:optimal-energy-general-equation} over all admissible sequences. As $\mathcal{T}_{\infty}$ is a \emph{strictly} concave function, the maximizing sequence is unique. The unique maximizer $\{\xi_j^*\}_{j=1}^\infty$ indeed achieves $\mathcal{T}_{\infty}^*$. Hence, \eqref{eq:Theorem:Average-Throughput} is justified. To investigate sequence $\{\xi_j^*\}_{j=1}^\infty$ as the unique maximizer of $\mathcal{T}_{\infty}$, we need to define the corresponding $N-$dimensional optimization problem in the following subsection to employ Karush-Kuhn-Tucker (KKT) conditions. Then, we investigate the relation between the finite dimensional optimization problem and the infinite dimensional optimization problem.\\
\begin{figure}[ht!]
  \centering
  \includegraphics[width = 8.5cm]{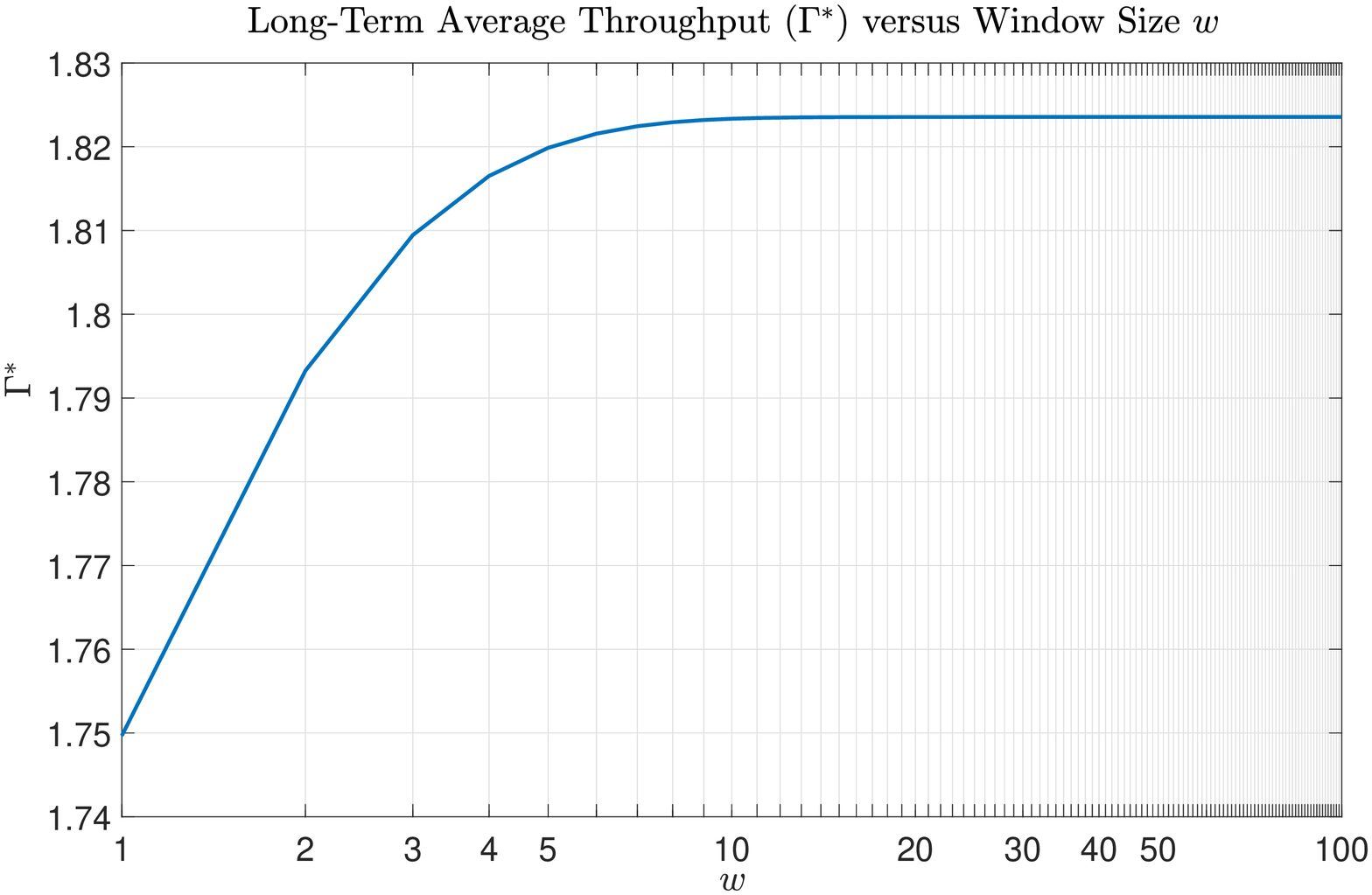}
  \caption{Average throughput $\Gamma^*$ versus look-ahead window size $w$ for fixed channel parameters $p = 0.3$ and $\gamma= 0.5$ and battery size $B = 100$.}\label{fig:Throughput-window}
\end{figure}
\subsection{$N-$Dimensional Optimization Problem.}
\indent  Let define the following $N-$dimensional optimization problem based on Definition~\ref{def:Theta_inf}.
\begin{definition}
  Fix $N\in \mathds{N}$. Recall Definition~\ref{def:Theta_inf}, and set
  $x_j = \xi_j$ for $j\in \{1,\ldots, N\}$ and $x_j = 0$ for $j > N$. Define
  \begin{align*}
    &\mathcal{T}_N(\xi_1, \cdots, \xi_N) \triangleq \mathcal{T}_\infty(\xi_1, \cdots, \xi_N, 0, 0, \cdots) \nonumber \\
    &\quad= \sum\limits_{k=1}^wp^2(1-p)^{k-1}\frac{k}{2}\log(1+\gamma\frac{B}{k}) \nonumber\\
    &\qquad+\sum\limits_{j=1}^N p(1-p)^{j+w-1}\frac{1}{2}\log(1+\gamma \xi_j) \nonumber\\
&\qquad+\sum\limits_{k=1}^{N}
p^2(1-p)^{k+w-1}\frac{w}{2}\log(1+\gamma\frac{B-\sum_{j=1}^{k}\xi_j}{w}) \nonumber\\
&\qquad+ p(1-p)^{w+N}\frac{w}{2}\log(1+\gamma\frac{B-\sum_{j=1}^{N}\xi_j}{w})\:. \label{optimize-target-finite}
  \end{align*}
The $N-$\emph{dimensional optimization problem} is defined as
\begin{equation}\label{eq:N-dimensional-optimization}
    \mathcal{T}_N^* = \sup \mathcal{T}_N(\xi_1^{(N)}, \cdots, \xi_N^{(N)})\:,
\end{equation}
where the supremum is over all sequences $\{\xi_j^{(N)}\}_{j=1}^N$ subject to
\begin{subequations}
\begin{eqnarray}
\xi_j^{(N)} \geq 0\:,\\
\sum_{j=1}^{N} \xi_j^{(N)} \leq B\:. \label{eq:N-dimensional-optimization:energy-constraint}
\end{eqnarray}
\end{subequations}
The maximizer of~$\mathcal{T}_N$, if exists, is denoted by $\{\xi_j^{(N)*}\}_{j=1}^N$.
\end{definition}
\begin{lem} \label{Lem:Limits-of-target-function-LB}
The following statements are valid for $\mathcal{T}_N^*$.
\begin{enumerate}[(a).]
  \item There exists a unique maximizer $\{\xi_j^{(N)*}\}_{j=1}^N$ of $\mathcal{T}_N$.
  \item $\mathcal{T}_N^*$ is an increasing function of $N$.
  \item $\lim\limits_{N \rightarrow \infty} \mathcal{T}_N^* = \mathcal{T}_\infty^*$ and it is finite.
\end{enumerate}
\end{lem}
\begin{proof}
(a) follows the fact that function $\mathcal{T}_N$ is a continuous bounded concave function defined on a compact set. (b) is due to the fact that domain of $\mathcal{T}_N$ is a subset of domain of $\mathcal{T}_{N+1}$. (c) follows from the
fact that $\mathcal{T}_N$ is a continuous bounded function and is increasing, i.e., $\mathcal{T}_N^* \leq \mathcal{T}_\infty^*$, and thus $\lim\limits_{N\rightarrow \infty}\mathcal{T}_N^*$ exists and it is finite. Indeed, the following inequality proves that this limit is $\mathcal{T}_\infty^*$ due to the  squeeze theorem.
\begin{align}
    &\mathcal{T}_N^*(\xi_1^{(N)*}, \cdots, \xi_N^{(N)*}) \nonumber\\
    &\quad\geq \mathcal{T}_N(\xi_1^*, \cdots, \xi_N^*) \nonumber \\
    &\quad=  \mathcal{T}_\infty^*(\xi_1^*, \cdots, \xi_\infty^*) - \sum_{j= N+1}^{\infty} \frac{p(1-p)^{j+w-1}}{2} \log(1+ \gamma \xi_j^*)\nonumber\\
     &\qquad  - \sum_{k= N+1}^{\infty} p^2 (1-p)^{k+w-1} \frac{w}{2} \log(1+\frac{\gamma}{w} (B- \sum_{j=1}^{k}\xi_j^*))\nonumber\\
     &\qquad + p(1-p)^{w+N}\frac{w}{2}\log(1+\gamma\frac{B-\sum_{j=1}^{N}\xi_j^*}{w})\nonumber\\
    &\quad \geq \mathcal{T}_\infty^* - \sum_{j= N+1}^{\infty} \frac{p(1-p)^{j+w-1}}{2} \log(1+ \gamma B) \nonumber\\
    &\qquad  - \sum_{k= N}^{\infty} p^2 (1-p)^{k+w-1} \frac{w}{2} \log(1+\frac{\gamma}{w} B) \nonumber\\
    &\quad = \mathcal{T}_\infty^* - \frac{(1-p)^{w+N}}{2} \left[\log(1+ \gamma B) + pw \log(1+\frac{\gamma}{w} B)\right] \nonumber\\
    &\quad = \mathcal{T}_\infty^* - \epsilon_N \:,  \label{eq:Lem:Target-Function-limits}
\end{align}
where $\epsilon_N \triangleq \frac{(1-p)^{w+N}}{2} \left[\log(1+ \gamma B) + pw \log(1+\frac{\gamma}{w} B)\right]$ is a positive number with property $\lim\limits_{N\rightarrow \infty}\epsilon_N = 0$.
\end{proof}
\indent In Corollary~\ref{corol:Theta_N:strictly-Inc} in the sequel, we will prove that $\mathcal{T}_N^*$ is indeed a \emph{strictly} increasing function, which is stronger than Lemma~\ref{Lem:Limits-of-target-function-LB}-part(B).\\
\indent The $N$-dimensional optimization problem can be solved by the KKT method. A \emph{necessary and sufficient} condition for sequence $\{\xi_j^{(N)*}\}_{j=1}^N$ to attain $\mathcal{T}^*_N$ is to satisfy
\begin{subequations} \label{finit:kkt}
\begin{align}
&p(1-p)^{w+j-1}\frac{\gamma}{2(1+\gamma
\xi_j^{(N)*})}\nonumber\\
&\:\: -\sum\limits_{k=j}^{N-1}
p^2(1-p)^{k+w-1}\frac{\frac{\gamma}{2}}{1+\gamma\frac{B-\sum_{i=1}^{k}
\xi_i^{(N)*}}{w}} \nonumber\\
&\:\: -p(1-p)^{w+N-1}\frac{\frac{\gamma}{2}}{1+\gamma\frac{B-\sum_
{i=1}^{N} \xi_i^{(N)*}}{w}}-\lambda^{(N)}+\mu^{(N)}_j=0\:, \label{finit:kkt1}\\
&\mu^{(N)}_j \xi_j^{(N)*}=0\:, \label{finit:kkt2}\\
&\lambda^{(N)}(B-\sum\limits_{i=1}^N \xi_i^{(N)*})=0\:,\label{finit:kkt3}
\end{align}
\end{subequations}
where $j \in \{1, 2, \cdots, N\}$, and $\lambda^{(N)}, \{\mu^{(N)}_i\}_{i=1}^N$ are non-negative real numbers.\\
\subsection{Properties of Sequence $\{\xi_j^{(N)*}\}_{j=1}^N$.}
In this subsection, we investigate the properties of the sequence $\{\xi_j^{(N)*}\}_{j=1}^N$ which achieves  $\mathcal{T}_N^*$ based on the KKT conditions~\eqref{finit:kkt}.
\begin{definition}
  The effective length of an admissible energy sequence is the largest (time) index beyond which the rest of the sequence vanishes. Specifically, $N_{eff}$ is called the effective length of $\{\xi_j^{(N)*}\}_{j=1}^N$ if
  \begin{align} \label{Eq:def:Effective-length}
  &\xi_{N_{eff}} > 0 \\
  &\xi_j^{(N)*} = 0 \quad \text{for all} \quad\: j > N_{eff}.
\end{align}
\end{definition}

\begin{lem} \label{lem:largest_nonzero}
The effective length of sequence $\{\xi_j^{(N)*}\}_{j=1}^N$ is $N$.
\end{lem}
\begin{proof}
First, we show that there exists at least one non-zero elements in the sequence of $\{\xi_j^{(N)*}\}_{j=1}^N$. Assume that this claim is not valid, and an all zero sequence is the optimal solution which satisfies the KKT conditions~\eqref{finit:kkt}. Thus,
\begin{align}
&\xi_j^{(N)*} = 0 \quad \text{for} \quad j\in \{1, \cdots, N\} \label{Eq:lem:nonzero-property-2}\\
&\sum_{i=1}^{N} \xi_i^{(N)*} = 0 \label{Eq:lem:nonzero-property-3}
\end{align}

Hence, $\lambda^{(N)} = 0$ due to \eqref{finit:kkt3}. From $\lambda^{(N)} = 0$, $\mu_{N}\geq 0$, \eqref{finit:kkt1} (when $j=N$ is set) and \eqref{Eq:lem:nonzero-property-3}, we should have
\begin{align}
p(1-p)^{w+N-1}\frac{\gamma}{2(1+\gamma
\xi_N^{(N)*})}-p(1-p)^{w+N-1}\frac{\frac{\gamma}{2}}{1+\gamma\frac{B}{w}}\leq 0
\end{align}
However, this equation is only valid when $\xi_{N}\geq \frac{\gamma}{w}B > 0$ due to $B>0$. This inequality contradicts with \eqref{Eq:lem:nonzero-property-2}.
Consequently, the largest (time) index of the non-zero element ($J$) exists with property~\eqref{Eq:def:Effective-length}.\\
Second, we prove that $J < N$ is not valid in the following.
Suppose $J<N$. Define $B'=\sum_{j=1}^J \xi_j^{(N)*}$.  Note
that we have
\begin{align}\label{proof:eq:constraint_mltiplyers}
  \lambda^{(N)} &\geq 0 \\
  \mu^{(N)}_J &=0
\end{align}
where the second inequality is due to $\xi_J^{(N)*}>0$ and~\eqref{finit:kkt2}.
First, assume $B=B'$; Setting $j=J$ in~\eqref{finit:kkt1} gives
\begin{align*}
\lambda^{(N)} = -p(1-p)^{w+J-1}\frac{\gamma^2 \xi_J^{(N)*}}{2(1+\gamma \xi_J^{(N)*})},
\end{align*}
which is contradictory with the fact that $\lambda^{(N)}\geq
0$. Next consider the case $B'<B$. In this case, $\lambda^{(N)} = 0$. Setting $j=J+1$ in~\eqref{finit:kkt1} gives
\begin{align}
[p(1-p)^{w+J}\frac{\gamma}{2}-p(1-p)^{w+J}\frac{\frac{\gamma}{2}}{1+\gamma\frac{B-B'}{w}}]+\mu_{J+1}=0.\label{eq:contradiction}
\end{align}
However, in~\eqref{eq:contradiction}, the bracket is strictly positive. Hence, the left hand side of~\eqref{eq:contradiction} is strictly positive. This lead to a contradiction. Therefore, $J=N$ always holds.
\end{proof}

\begin{corol} \label{corol:nonzero_last_term}
The last element of the optimal sequence $\{\xi_i^{(N)*}\}_{i=1}^N$ is always non-zero, i.e., $\xi_N^{(N)*} > 0$ and coefficient $\mu^{(N)}_N=0$.
\end{corol}
\begin{proof}
Due to Lemma~\ref{lem:largest_nonzero}, the largest non-zero element is the $N^{th}$ element of the sequence. $\mu^{(N)}_N=0$ is due to~\eqref{finit:kkt2}.
\end{proof}

\begin{corol} \label{corol:Theta_N:strictly-Inc}
$\mathcal{T}_N^*$ is a \emph{strictly} increasing function of $N$.
\end{corol}
\begin{proof}
\begin{eqnarray*}
\mathcal{T}_N^*(\xi_1^{(N)*}, \cdots, \xi_N^{(N)*}) &=& \mathcal{T}_{N+1}(\xi_1^{(N)*}, \cdots, \xi_N^{(N)*}, 0)\nonumber\\
&<& \mathcal{T}_{N+1}^*(\xi_1^{(N+1)*}, \cdots, \xi_N^{(N+1)*}, \xi_{N+1}^{(N+1)*})
\end{eqnarray*}
where the last inequality is due to the fact that $(\xi_1^{(N)*}, \cdots, \xi_N^{(N)*}, 0)$ can not be the unique maximizer of $\mathcal{T}_{N+1}$, because $\xi_{N+1}^{(N+1)*}>0$ due to Corollary~\ref{corol:nonzero_last_term}.
\end{proof}

\begin{lem} \label{lem:Inequality}
For any $j \in \{1, 2, \cdots, N-1\}$, we have
\begin{equation}
  \xi_j^{(N)*} \leq \frac{B- \sum_{i=1}^{j}\xi_i^{(N)*}}{w}\:, \label{eq:lem:Inequality-1}
\end{equation}
where inequality holds in the \emph{strict sense} for $j < N$.
\end{lem}
\begin{proof}
If $\xi_j^{(N)*} = 0$, the inequality holds because $\sum_{i=1}^{j}\xi_i^{(N)*}<B$ according to~\eqref{eq:N-dimensional-optimization:energy-constraint} and Corollary~\ref{corol:nonzero_last_term} when $j<N$. Otherwise, if $\xi_j^{(N)*} >0$, then $\mu^{(N)}_j =0$ due to \eqref{finit:kkt2}. Hence, we can derive the expressions \eqref{eq:proof:dec1}-\eqref{eq:proof:dec4} at the bottom of this page, where \eqref{eq:proof:dec1} is due to \eqref{finit:kkt1} and $\lambda^{(N)} \geq 0$, \eqref{eq:proof:dec2} follows from $B-\sum_{i=1}^{k-w}
\xi_i^{(N)*} \leq B-\sum_{i=1}^{j}
\xi_i^{(N)*}$, for $k = w+j+1, \cdots, w+N-1$, due to~\eqref{eq:N-dimensional-optimization:energy-constraint}, \eqref{eq:proof:dec3} holds \emph{strictly} only if $j<N$, because $B-\sum_{i=1}^{N}
\xi_i^{(N)*} < B-\sum_{i=1}^{j}
\xi_i^{(N)*}$ due to Corollary \ref{corol:nonzero_last_term} for $j<N$. Therefore, the lemma is concluded from~\eqref{eq:proof:dec4}.
\end{proof}

\begin{figure*}[!b]
\hrulefill
\setcounter{MYtempeqncnt}{\value{equation}}
\setcounter{equation}{33}
\begin{align}
p(1-p)^{w+j-1}\frac{\frac{\gamma}{2}}{1+\gamma
\xi_j^{(N)*}}&\geq \sum\limits_{k=j}^{N-1}
p^2(1-p)^{k+w-1}\frac{\frac{\gamma}{2}}{1+\gamma\frac{B-\sum_{i=1}^{k}
\xi_i^{(N)*}}{w}}+ p(1-p)^{w+N-1}\frac{\frac{\gamma}{2}}{1+\gamma\frac{B-\sum_
{i=1}^{N}\xi_i^{(N)*}}{w}} \label{eq:proof:dec1}\\
&\geq \sum\limits_{k=j}^{N-1}
p^2(1-p)^{k+w-1}\frac{\frac{\gamma}{2}}{1+\gamma\frac{B-\sum_{i=1}^{j}
\xi_i^{(N)*}}{w}}+ p(1-p)^{w+N-1}\frac{\frac{\gamma}{2}}{1+\gamma\frac{B-\sum_
{i=1}^{N}\xi_i^{(N)*}}{w}} \label{eq:proof:dec2}\\
&= p(1-p)^{w+j-1}(1-(1-p)^{(N-j)})\frac{\frac{\gamma}{2}}{1+\gamma\frac{B-\sum_{i=1}^{j}
\xi_i^{(N)*}}{w}}
+ p(1-p)^{w+N-1}\frac{\frac{\gamma}{2}}{1+\gamma\frac{B-\sum_
{i=1}^{N}\xi_i^{(N)*}}{w}} \nonumber\\
&\geq p(1-p)^{w+j-1}(1-(1-p)^{(N-j)})\frac{\frac{\gamma}{2}}{1+\gamma\frac{B-\sum_{i=1}^{j}
\xi_i^{(N)*}}{w}}
+ p(1-p)^{w+N-1}\frac{\frac{\gamma}{2}}{1+\gamma\frac{B-\sum_
{i=1}^{j}\xi_i^{(N)*}}{w}} \label{eq:proof:dec3}\\
&= p(1-p)^{w+j-1} \frac{\frac{\gamma}{2}}{1+\gamma\frac{B-\sum_{i=1}^{j}
\xi_i^{(N)*}}{w}}\:, \label{eq:proof:dec4}
\end{align}
\setcounter{equation}{37}
\end{figure*}

\indent For any fixed $N$, let define parameter
\begin{equation*}
  B^{(N)} = B- \sum_{i=1}^{N} \xi_i^{(N)*}\:.
\end{equation*}
From Corollary~\ref{corol:nonzero_last_term} and Lemma~\ref{lem:Inequality}, we conclude that
\begin{equation} \label{eq:positive-remain-battery}
  B^{(N)} > 0\:.
\end{equation}
Also, from~\eqref{finit:kkt3} and~\eqref{eq:positive-remain-battery}, we conclude that
\begin{equation}\label{eq:lambda_zero}
  \lambda^{(N)} = 0\:.
\end{equation}
\indent In the following lemma, we investigate the behaviour of sequence $\{\xi_j^{(N)*}\}_{j=1}^N$ as a function of time index $j$ when $N$ is fixed.
\begin{lem} \label{lem:decrease:j}
$\{\xi_j^{(N)*}\}_{j=1}^N$ is a strictly decreasing positive sequence.
\end{lem}
\begin{proof}
From~\eqref{finit:kkt1}, for two successive terms $j<N$ and $j+1$, we have
\begin{align*}
&\frac{p(1-p)^{w+j-1} (\frac{\gamma}{2})}{1+\gamma
\xi_j^{(N)*}}-\sum\limits_{k=j}^{N-1}
\frac{p^2(1-p)^{w+k-1} (\frac{\gamma}{2})}{1+\gamma\frac{B-\sum_{i=1}^{k}
\xi_i^{(N)*}}{w}}\nonumber\\
&-\frac{p(1-p)^{w+N-1}(\frac{\gamma}{2})}{1+\gamma\frac{B-\sum_
{i=1}^{N} \xi_i^{(N)*}}{w}}-\lambda^{(N)}+\mu^{(N)}_j=0
\end{align*}
and
\begin{align*}
&\frac{p(1-p)^{w+j} (\frac{\gamma}{2})}{1+\gamma
\xi_{j+1}^{(N)*}}-\sum\limits_{k=j+1}^{N-1}
\frac{p^2(1-p)^{w+k-1} (\frac{\gamma}{2})}{1+\gamma\frac{B-\sum_{i=1}^{k}
\xi_i^{(N)*}}{w}}\nonumber\\
&-\frac{p(1-p)^{w+N-1}(\frac{\gamma}{2})}{1+\gamma\frac{B-\sum_
{i=1}^{N} \xi_i^{(N)*}}{w}}-\lambda^{(N)}+\mu^{(N)}_{j+1}=0
\end{align*}
respectively. Subtracting the first equation from the second one, we obtain
\begin{eqnarray}
\frac{p(1-p)^{w+j}(\frac{\gamma}{2})}{1+\gamma \xi_{j+1}^{(N)*}} &=& \frac{p(1-p)^{w+j-1}(\frac{\gamma}{2})}{1+\gamma
\xi_{j}^{(N)*}}\nonumber\\
&&- \frac{p^2(1-p)^{w+j-1}(\frac{\gamma}{2})}{1+\gamma\frac{B-\sum_{i=1}^{j}
\xi_i^{(N)*}}{w}} - \mu^{(N)}_{j+1} + \mu^{(N)}_j \nonumber\\
\label{eq:successive_terms-0}
\end{eqnarray}
From Lemma~\ref{lem:Inequality} (for case $j<N$), we have $\frac{\frac{\gamma}{2}}{1+\gamma\frac{B-\sum_{i=1}^{j}
\xi_i^{(N)*}}{w}} < \frac{\frac{\gamma}{2}}{1+\gamma \xi_j^{(N)*}}$. Hence, \eqref{eq:successive_terms-0} can be simplified into
\begin{equation} \label{eq:successive_terms}
\frac{p(1-p)^{w+j}(\frac{\gamma}{2})}{1+\gamma \xi_{j+1}^{(N)*}} > \frac{p(1-p)^{w+j}(\frac{\gamma}{2})}{1+\gamma
\xi_{j}^{(N)*}} - \mu^{(N)}_{j+1} + \mu^{(N)}_j
\end{equation}
First, note that $\xi_{N}^{(N)*}> 0$ due to Corollary~\ref{corol:nonzero_last_term}. Second, start with $j = N-1$. In this case $\mu^{(N)}_j\geq 0$ and $\mu^{(N)}_{j+1} = 0$ due to Corollary~\ref{corol:nonzero_last_term}. Hence, from~\eqref{eq:successive_terms}, we conclude that $\xi_{j+1}^{(N)*}< \xi_{j}^{(N)*}$. Therefore, $\xi_{j}^{(N)*}> 0$ due to Corollary~\ref{corol:nonzero_last_term}, and thus $\mu^{(N)}_j = 0$ due to~\eqref{finit:kkt2}. Then, repeat this justification for $j=N-2, N-3, \ldots, 1$ in the descending order to establish this lemma for all $j\in\{1, \cdots, N-1\}$.
\end{proof}
\begin{remark}
Note that Lemma~\ref{lem:decrease:j} is different from Lemma~\ref{lem:largest_nonzero} regarding positivity property of each element. This lemma only guarantees that a \emph{finite} number of zero elements in sequence $\{\xi_j^{(N)*}\}_{j=1}^N$ exists, while Lemma~\ref{lem:decrease:j} guarantees no zero element exists and thus it is stronger. As a result of Lemma~\ref{lem:decrease:j}, it can be concluded that
\begin{equation}\label{eq:zero-mu}
  \mu^{(N)}_j = 0
\end{equation}
holds for any $j\in\{1, 2, \ldots, N\}$ in~\eqref{finit:kkt2}.
\end{remark}
\indent In the following lemma, we establish a recursive expression to simplify the KKT conditions and to obtain energy sequence $\{\xi_j^{(N)*}\}_{j=1}^N$, recursively.
\begin{lem} \label{Lemma:KKT:recursive}
The following identities hold for energy sequence $\{\xi_j^{(N)*}\}_{j=1}^N$, which satisfies the KKT conditions~\eqref{finit:kkt}.
\begin{subequations}
\begin{align}
  \xi_N^{(N)*} &= \frac{B^{(N)}}{w}\:. \label{eq:last_element_reside_equality}\\
    \frac{1-p}{1+\gamma \xi_{j+1}^{(N)*}} &= \frac{1}{1+\gamma \xi_j^{(N)*}} - \frac{p}{1+\frac{\gamma}{w}(B-\sum_{i=1}^{j} \xi_i^{(N)*})} \label{eq:KKT:recursive}
\end{align}
for $j\in\{1, 2, \ldots, N-1\}$.
\end{subequations}
\end{lem}
\begin{proof}
The KKT condition~\eqref{finit:kkt1} can be simplified due to~\eqref{eq:lambda_zero} and~\eqref{eq:zero-mu}. For $j=N$, we have
\begin{equation*}
\frac{p(1-p)^{w+N-1}(\frac{\gamma}{2})}{1+\gamma \xi_N^{(N)*}} = \frac{p(1-p)^{w+N-1}(\frac{\gamma}{2})}{1+\gamma\frac{B-\sum_
{i=1}^{N}\xi_i^{(N)*}}{w}}
\end{equation*}
from which~\eqref{eq:last_element_reside_equality} is derived. For $j\in\{1, 2, \ldots, N-1\}$, \eqref{eq:KKT:recursive} is derived from~\eqref{eq:successive_terms-0} and~\eqref{eq:zero-mu}.

\end{proof}
\begin{remark}
From Lemma~\ref{Lemma:KKT:recursive}, we conclude that the equality in~\eqref{lem:Inequality} for $j=N$ always holds.
\end{remark}
\indent In the following lemma, we investigate the behaviour of sequence $\{\xi_j^{(N)*}\}_{N=j}^\infty$ as a function of $N$ for a fixed time instant $j$.
\begin{lem} \label{lem:decrease:N}
Consider sequence $\{\xi_j^{(N)*}\}_{N=j}^\infty$ for a given fixed $j\in \mathds{N}$. Then, $\{\xi_j^{(N)*}\}_{N=j}^\infty$ is a \emph{strictly} decreasing function of $N$.
\end{lem}
\begin{proof}
Let consider two different assumptions for $\xi_1^{(N)*}$ and $\xi_1^{(N+1)*}$. First, assume that $\xi_1^{(N+1)*} \geq \xi_1^{(N)*}$. If this is the case, from~\eqref{eq:KKT:recursive}, we can deduce
\begin{equation} \label{eq:UB-decrease-N:contradic0}
  \xi_j^{(N+1)*} \geq \xi_j^{(N)*} \qquad \text{for any} \qquad j\in\{1,2, \cdots, N\}.
\end{equation}
However, this conclusion contradicts the following inequality for $j = N$.
\begin{eqnarray}
  \xi_N^{(N+1)*} &<& \frac{B- \sum_{i=1}^{N} \xi_i^{(N+1)*}}{w} \label{eq:UB-decrease-N:contradic1}\\
               &\leq& \frac{B- \sum_{i=1}^{N} \xi_i^{(N)*}}{w} \label{eq:UB-decrease-N:contradic2}\\
               &=& \xi_N^{(N)*} \label{eq:UB-decrease-N:contradic3}\\
\end{eqnarray}
where~\eqref{eq:UB-decrease-N:contradic1} follows from~\eqref{eq:lem:Inequality-1} (of Lemma~\ref{lem:Inequality}), \eqref{eq:UB-decrease-N:contradic2} follows from~\eqref{eq:UB-decrease-N:contradic0}, and \eqref{eq:UB-decrease-N:contradic3} follows from~\eqref{eq:last_element_reside_equality}. Hence, the assumption $\xi_1^{(N+1)*} < \xi_1^{(N)*}$ is correct. Therefore, from~\eqref{eq:KKT:recursive}, we can deduce
\begin{equation*}
  \xi_j^{(N+1)*} < \xi_j^{(N)*} \qquad \text{for any} \qquad j\in\{1,2, \cdots, N\}.
\end{equation*}
\end{proof}
\begin{lem} \label{Lem:deacrease:battery-rem-LB}
Sequence $\{B^{(N)}\}_{N=1}^\infty$ is a \emph{strictly} decreasing function of $N$.
\end{lem}
\begin{proof}
From Lemma~\ref{lem:decrease:N}, we conclude that $\xi_N^{(N+1)*} < \xi_N^{(N)*}$. On the other hand, from Lemma~\ref{lem:decrease:j}, we conclude that $\xi_{N+1}^{(N+1)*} < \xi_N^{(N+1)*}$. Hence, inequality $\xi_{N+1}^{(N+1)*} < \xi_N^{(N)*}$ leads to inequality $B^{(N+1)} < B^{(N)}$ due to~\eqref{eq:last_element_reside_equality}.
\end{proof}

\begin{corol} \label{corol:B_rem-limit:LB}
The limit of sequence $\{B^{(N)}\}_{N=1}^\infty$ exists as follows.
      \begin{equation*}
        \lim_{N\rightarrow \infty} B^{(N)} = 0\: .
      \end{equation*}
\end{corol}
\begin{proof}
Sequence $\{B^{(N)}\}_{N=1}^\infty$ is a strictly decreasing function of $N$, due to Lemma~\ref{Lem:deacrease:battery-rem-LB}, and positive due to~\eqref{eq:positive-remain-battery}. Hence, it converges to a non-negative real number. To find the limit, we can bound $B^{(N)}$ as follows.
\begin{eqnarray}
  B^{(N)} &=& w \xi^{(N)*}_{N} \label{Eq:Lem:Brem_UB-1}\\
   &<& \frac{w}{N} (\sum_{i=1}^{N} \xi^{(N)*}_{i}) \label{Eq:Lem:Brem_UB-2}\\
   &\leq& \frac{w}{N} B \label{Eq:Lem:Brem_UB-3}\:,
\end{eqnarray}
where \eqref{Eq:Lem:Brem_UB-1} follows from~\eqref{eq:last_element_reside_equality}, \eqref{Eq:Lem:Brem_UB-2} follows from Lemma~\ref{lem:decrease:j}, and \eqref{Eq:Lem:Brem_UB-3} follows from energy constraint~\eqref{eq:N-dimensional-optimization:energy-constraint}.
The limit is concluded from \eqref{Eq:Lem:Brem_UB-3}.
\end{proof}
\subsection{Properties of sequence $\{\xi_j^*\}_{j=1}^\infty$.}
\begin{lem} \label{Lem:limits:LB}
The limits of sequence $\{\xi_j^{(N)*}\}_{N=j}^\infty$ for any $j\in \mathds{N}$, where $j$ is fixed, exists as follows.
      \begin{equation}\label{eq:E_j^*}
        \xi_j^* = \lim_{N\rightarrow \infty} \xi_j^{(N)*} \qquad \text{for} \qquad j\in\{1, \ldots, N\},
      \end{equation}
      where $\xi_j^*$ is element $j^{th}$ of the maximizing sequence of $\mathcal{T}_\infty$.
\end{lem}
\begin{proof}
Sequence $\{\xi_j^{(N)*}\}_{N=j}^\infty$, for any fixed $j\in \mathds{N}$, consists of positive elements due to Lemma~\eqref{lem:decrease:j} and it strictly descends as a function of $N$ due to Lemma~\ref{lem:decrease:N}. Hence, for each $j\in \mathds{N}$, there exits a finite number $\alpha_j$ such that
\begin{equation} \label{eq:Lem:limits-E_j^(N)*:LB-0}
  \alpha_j = \lim_{N\rightarrow \infty} \xi_j^{(N)*}\:.
\end{equation}
In the following, we prove that $\{\alpha_j\}_{j=1}^\infty$ is the unique maximizer of $\mathcal{T}_\infty$. For any fixed $N_0$, we have
\begin{eqnarray}
\mathcal{T}_\infty(\alpha_1, \alpha_2, \cdots) &=& \lim_{N\rightarrow \infty} \mathcal{T}_N(\xi_1^{(N)*}, \xi_2^{(N)*}, \cdots, \xi_N^{(N)*}) \label{eq:Lem:limits-E_j^(N)*:LB-1}\\
&\geq& \mathcal{T}_{N_0}^*(\xi_1^{(N_0)*}, \xi_2^{(N_0)*}, \cdots, \xi_{N_0}^{(N_0)*}) \label{eq:Lem:limits-E_j^(N)*:LB-2}\\
&\geq& \mathcal{T}_\infty^*(\xi_1^*, \xi_2^*, \cdots) - \epsilon_{N_0} \label{eq:Lem:limits-E_j^(N)*:LB-3}
\end{eqnarray}
where~\eqref{eq:Lem:limits-E_j^(N)*:LB-1} follows from Lemma~\ref{Lem:Limits-of-target-function-LB}-(part C) and \eqref{eq:Lem:limits-E_j^(N)*:LB-0}, \ref{eq:Lem:limits-E_j^(N)*:LB-2} follows from Lemma~\ref{Lem:Limits-of-target-function-LB}-(part B), \eqref{eq:Lem:limits-E_j^(N)*:LB-3} follows from~\eqref{eq:Lem:Target-Function-limits}.
Now, if $N_0 \rightarrow \infty$, then $\epsilon_{N_0} \rightarrow 0$ due to~\eqref{eq:Lem:Target-Function-limits}, and thus \eqref{eq:Lem:limits-E_j^(N)*:LB-3} contradicts with Lemma~\ref{Lem:Limits-of-target-function-LB}-(part A) unless $\alpha_j = \xi_j^*$, or equivalently,
\begin{equation}\label{eq:Lem:limits-E_j^(N)*:LB-main}
  \xi_j^* = \lim_{N\rightarrow \infty} \xi_j^{(N)*}\:.
\end{equation}
\end{proof}
\indent Eventually, in Theorem~\ref{Theorem:Main-Result}, \eqref{eq:E_j^*_recursive} follows directly from~\eqref{eq:KKT:recursive} (of Lemma~\ref{Lemma:KKT:recursive}) and \eqref{eq:Lem:limits-E_j^(N)*:LB-main} (of Lemma~\ref{Lem:limits:LB}). Also, the energy constraint~\eqref{eq:energy_constraint-infinit} follows from Corollary~\ref{corol:B_rem-limit:LB} and~\eqref{eq:Lem:limits-E_j^(N)*:LB-main}. Also, \eqref{eq:theorem:property} follows from Lemma~\ref{lem:Inequality} and~\eqref{eq:Lem:limits-E_j^(N)*:LB-main}. The fact that $\{\xi_j^*\}_{j=1}^N$ is a strictly decreasing sequence follows from
    \begin{align}
    \frac{1-p}{1+\gamma \xi^*_j}&=\frac{1}{1+\gamma \xi^*_{j-1}}-\frac{p}{1+\frac{\gamma}{w}(B-\sum_{i=1}^{j-1} \xi^*_i)} \label{eq:proof:dec-E^*_j-1}\\
    &>\frac{1}{1+\gamma \xi^*_{j-1}}-\frac{p}{1+\gamma \xi^*_{j-1}} \label{eq:proof:dec-E^*_j-2}\\
    &=\frac{1-p}{1+\gamma \xi^*_{j-1}} \nonumber,
    \end{align}
    where~\eqref{eq:proof:dec-E^*_j-1} follows from~\eqref{eq:E_j^*_recursive}, and \eqref{eq:proof:dec-E^*_j-2} follows from~\eqref{eq:theorem:property}. This monotone property also approves the positivity of sequence $\{\xi_j^*\}_{j=1}^N$.\\

\indent In Fig.~\ref{fig: optimal-Seq}, $\{\xi_j^{(N)*}\}_{j=1}^N$ (following by trailing zeros) is sketched over time for different values of $N$, when the parameters of the system are as given. This figure numerically illustrates the results of this section.

\begin{remark}
After an energy arrival, the optimal policy $\pi^*(w)$ assigns a strictly positive energy sequence $\{\xi_j^*\}_{j=1}^\infty$ to time slots according to Algorithm~\ref{Fig:Algorithm} until an arrival is observed in the window. That is, as long as no energy arrival is observed in the window, the TX always sends a positive portion of the battery energy at each time slot, but it never exhausts the battery energy. However, the online optimal policies~\cite{ozgur16, Ulukus2017-ISIT, arafa2018} allocate positive energy only to a fixed number of time slots after any energy arrival and then the battery becomes depleted.
\end{remark}
\section{Conclusions}  \label{Sec:Conclusion}
In this paper, we have introduced a new EH system model, where the TX has the knowledge of future energy arrivals in a look-ahead window. This model provides a bridge between the online model and the offline model, which have been previously studied in isolation. A complete characterization of the optimal power control policy is obtained for the new model with Bernoulli energy arrivals, which has notable differences with its counterpart for the online model.
\begin{figure}[ht!]
  \centering
  \includegraphics[width = 9 cm]{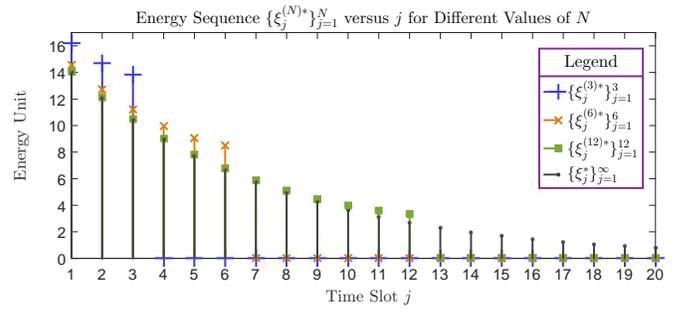}
  \caption{Sequence $\xi_j^{(N)*}$ versus time $j$ ($1\leq j \leq 20$) for model parameters $p = 0.3$, $\gamma= 0.5$,  battery size $B = 100$, and window size $w=4$.}\label{fig: optimal-Seq}
\end{figure}
\bibliographystyle{IEEEtran}
\bibliography{EH}

\begin{thebibliography}{10}
\providecommand{\url}[1]{#1}
\csname url@samestyle\endcsname
\providecommand{\newblock}{\relax}
\providecommand{\bibinfo}[2]{#2}
\providecommand{\BIBentrySTDinterwordspacing}{\spaceskip=0pt\relax}
\providecommand{\BIBentryALTinterwordstretchfactor}{4}
\providecommand{\BIBentryALTinterwordspacing}{\spaceskip=\fontdimen2\font plus
\BIBentryALTinterwordstretchfactor\fontdimen3\font minus
  \fontdimen4\font\relax}
\providecommand{\BIBforeignlanguage}[2]{{%
\expandafter\ifx\csname l@#1\endcsname\relax
\typeout{** WARNING: IEEEtran.bst: No hyphenation pattern has been}%
\typeout{** loaded for the language `#1'. Using the pattern for}%
\typeout{** the default language instead.}%
\else
\language=\csname l@#1\endcsname
\fi
#2}}
\providecommand{\BIBdecl}{\relax}
\BIBdecl

\bibitem{ulukus}
O.~Ozel and S.~Ulukus, ``Achieving {AWGN} capacity under stochastic energy
  harvesting,'' \emph{IEEE Transactions on Information Theory}, vol.~58,
  no.~10, pp. 6471--6483, 2012.

\bibitem{modiano}
M.~A. Zafer and E.~Modiano, ``A calculus approach to energy-efficient data
  transmission with quality-of-service constraints,'' \emph{IEEE/ACM
  Transactions on Networking (TON)}, vol.~17, no.~3, pp. 898--911, 2009.

\bibitem{yang2012}
J.~Yang and S.~Ulukus, ``Optimal packet scheduling in an energy harvesting
  communication system,'' \emph{IEEE Transactions on Communications}, vol.~60,
  no.~1, pp. 220--230, 2012.

\bibitem{tutuncuoglu2012}
K.~Tutuncuoglu and A.~Yener, ``Optimum transmission policies for battery
  limited energy harvesting nodes,'' \emph{IEEE Transactions on Wireless
  Communications}, vol.~11, no.~3, pp. 1180--1189, 2012.

\bibitem{yener}
O.~Ozel, K.~Tutuncuoglu, J.~Yang, S.~Ulukus, and A.~Yener, ``Transmission with
  energy harvesting nodes in fading wireless channels: Optimal policies,''
  \emph{IEEE Journal on Selected Areas in Communications}, vol.~29, no.~8, pp.
  1732--1743, 2011.

\bibitem{zhang12}
C.~K. Ho and R.~Zhang, ``Optimal energy allocation for wireless communications
  with energy harvesting constraints,'' \emph{IEEE Transactions on Signal
  Processing}, vol.~60, no.~9, pp. 4808--4818, 2012.

\bibitem{review15}
S.~Ulukus, A.~Yener, E.~Erkip, O.~Simeone, M.~Zorzi, P.~Grover, and K.~Huang,
  ``Energy harvesting wireless communications: A review of recent advances,''
  \emph{IEEE Journal on Selected Areas in Communications}, vol.~33, no.~3, pp.
  360--381, 2015.

\bibitem{ozgur16}
D.~Shaviv and A.~{\"O}zg{\"u}r, ``Universally near optimal online power control
  for energy harvesting nodes,'' \emph{IEEE Journal on Selected Areas in
  Communications}, vol.~34, no.~12, pp. 3620--3631, 2016.

\bibitem{Ulukus2017-ISIT}
A.~Arafa, A.~Baknina, and S.~Ulukus, ``Energy harvesting networks with general
  utility functions: Near optimal online policies,'' in \emph{Information
  Theory (ISIT), 2017 IEEE International Symposium on}.\hskip 1em plus 0.5em
  minus 0.4em\relax IEEE, 2017, pp. 809--813.

\bibitem{arafa2018}
------, ``Online fixed fraction policies in energy harvesting communication
  systems,'' \emph{IEEE Transactions on Wireless Communications}, vol.~17,
  no.~5, pp. 2975--2986, 2018.

\bibitem{zib-IWCIT}
A.~Zibaeenejad and P.~Parhizgar, ``Power management policies for slowly varying
  {B}ernoulli energy harvesting channels,'' in \emph{2018 Iran Workshop on
  Communication and Information Theory (IWCIT)}, 2018, pp. 1--6.

\bibitem{zib-IEMCON}
------, ``Optimal universal power management policies for channels with
  slow-varying harvested energy,'' in \emph{Information Technology, Electronics
  and Mobile Communication Conference (IEMCON), 2018 IEEE 9th Annual}.\hskip
  1em plus 0.5em minus 0.4em\relax IEEE, 2016, pp. 1--8.

\bibitem{Puterman}
M.~L. Puterman, \emph{Markov decision processes: discrete stochastic dynamic
  programming}.\hskip 1em plus 0.5em minus 0.4em\relax John Wiley \& Sons,
  2014.

\bibitem{survey}
A.~Arapostathis, V.~S. Borkar, E.~Fern{\'a}ndez-Gaucherand, M.~K. Ghosh, and
  S.~I. Marcus, ``Discrete-time controlled markov processes with average cost
  criterion: a survey,'' \emph{SIAM Journal on Control and Optimization},
  vol.~31, no.~2, pp. 282--344, 1993.

\bibitem{Ross_prob}
R.~Sheldon \emph{et~al.}, \emph{A first course in probability}, 10th~ed.\hskip
  1em plus 0.5em minus 0.4em\relax Pearson, 2018.

\end{thebibliography}
\end{document}